\documentclass[conference,letterpaper]{IEEEtran}

\usepackage[utf8]{inputenc} 
\usepackage[T1]{fontenc}
\usepackage{url}
\usepackage{ifthen}
\usepackage{cite}
\usepackage{hyperref}
\usepackage[cmex10]{amsmath} 

\usepackage{bm}
\usepackage{amssymb, amsthm, amsfonts, mathtools, marvosym}
\newtheorem{definition}{Definition}
\newtheorem{theorem}{Theorem}
\newtheorem{lemma}{Lemma}
\newtheorem{corollary}{Corollary}

\newtheorem{proposition}{Proposition}
\newtheorem{assumption}{Assumption}
\newtheorem{remark}{Remark}

\theoremstyle{plain}

\usepackage{color}
\usepackage{xcolor}
\usepackage{graphicx}
\usepackage{subcaption}
\usepackage{extarrows}

\newcommand{\part}[2]{\frac{\partial #1}{\partial #2}}

\newcommand{\X}{\mathcal{X}}
\newcommand{\Y}{\mathcal{Y}}

\newcommand{\R}{\mathbb{R}}
\newcommand{\E}{\mathbb{E}}

\newcommand{\pc}{\mathcal{P}}

\newcommand{\norm}[1]{\left\Vert{#1}\right\Vert}

\begin{document}
\title{A Revisit to Rate-distortion Theory\\
via Optimal Weak Transport}

\author{
 \IEEEauthorblockN{Jiayang Zou, Luyao Fan, Jiayang Gao and Jia Wang}
 \IEEEauthorblockA{Department of Electronic Engineering \\
  Shanghai Jiao Tong University\\
  Shanghai, China\\
                   Email: qiudao@sjtu.edu.cn, fanluyao@sjtu.edu.cn, gjy0515@sjtu.edu.cn, jiawang@sjtu.edu.cn}
                    }


\maketitle


\begin{abstract}
    This paper revisits the rate-distortion theory from the perspective of optimal weak transport, as recently introduced by Gozlan et al. While the conditions for optimality and the existence of solutions are well-understood in the case of discrete alphabets, the extension to abstract alphabets requires more intricate analysis. Within the framework of weak transport problems, we derive a parametric representation of the rate-distortion function, thereby connecting the rate-distortion function with the Schr\"odinger bridge problem, and establish necessary conditions for its optimality. As a byproduct of our analysis, we reproduce K. Rose's conclusions regarding the achievability of Shannon lower bound concisely, without reliance on variational calculus.


\end{abstract}

\section{Introduction}
Rate-distortion (RD) theory is a fundamental branch of information theory, primarily concerned with the problem of compressing data at the minimum coding rate subject to a given distortion constraint. Initially proposed by Claude E. Shannon in his groundbreaking works in 1948 \cite{shannon1948mathematical} and 1959 \cite{shannon1959coding}, this theory has since become a cornerstone of modern data compression and communication system design. 
RD problems focus on designing a coding scheme that minimizes the coding rate, while satisfying a specific distortion level for a given source and distortion measure.

The classical rate-distortion problems for finite-alphabet sources and continuous sources defined in Euclidean spaces have been extensively studied in \cite{berger2003rate,blahut1972computation,blahut1987principles,cover2006elements}. As research advanced, scholars began to explore more complex problems involving abstract alphabets. A key contribution was made by I. Csisz\'ar, who extended RD theory to abstract alphabets \cite{csiszar1974extremum}. He rigorously derived the parametric representation of the RD function, laying a solid theoretical foundation for numerical algorithms \cite{blahut1972computation,arimoto1972algorithm}.
Subsequently, F. Rezaei et al. \cite{rezaei2006rate} significantly advanced Csisz\'ar's results by employing techniques such as Stone-C\v{e}ch compactification and fixed-point methods in locally convex topological vector spaces. These approaches eliminated the assumption of compactness for the target space. Moreover, they provided a comprehensive characterization of the optimal reconstruction under these generalized conditions. 
Further contributions were made by E. Riegler et al. \cite{riegler2018rate, riegler2023lossy}, who focused on lossy compression theory for sources defined on compact manifolds and fractal sets. By leveraging results related to the rate-distortion dimension \cite{kawabata1994rate}, they derived the Shannon lower bound (SLB) in more general settings, thereby extending the applicability of RD theory to more complex source structures.
In terms of the parametric representation of the RD function, V. Kostina and E. Tuncel \cite{kostina2019successive} avoided reliance on the Karush-Kuhn-Tucker (KKT) optimality conditions. Instead, they streamlined I. Csisz\'ar's argument by utilizing the Donsker-Varadhan characterization of minimum relative entropy \cite{donsker1975asymptotic}. Their work significantly advanced the theory of successive refinement for abstract sources, thereby generalizing the earlier results of Equitz and Cover \cite{equitz1991successive}.

In 2013, Robert M. Gray highlighted the connections between optimal transport (OT) and information theory \cite{gray2013transportation}, offering a geometric view of source coding and a brief survey of optimal transport from a quantization perspective. Since then, the relationship between OT and RD theory has gradually garnered attention. However, conventional optimal transport theory is challenging to apply directly to RD problems due to fundamental differences in their problem settings. Conventional optimal transport minimizes transportation cost between probability distributions with marginal constraints, while RD theory balances coding rate and distortion under a disintegration framework. To bridge these differences, an entropy regularization term have been introduced to establish equivalence within the framework of entropic optimal transport in \cite{yang2024estimating}. In fact, optimal weak transport (OWT) theory, as a generalization of OT, naturally links to RD problems through optimizing the disintegration of the target measure with respect to the source measure. This framework enables a more concise reproduction of several key results in RD theory.

In this paper we conduct a study of revisiting the RD theory via OWT, which was recently introduced by Gozlan et al. \cite{gozlan2017kantorovich}. We derive a parametric representation that connects the RD function with the Schr\"odinger problem \cite{leonard2013survey} within the framework of weak transport problems. Additionally, we significantly streamline the proof of the necessary optimality conditions for the RD function using the existence theorem of OWT problems \cite{backhoff2019existence,backhoff2022applications}. This approach provides a concise method to reproduce K. Rose's conclusions regarding the achievability of the Shannon lower bound \cite{rose1994mapping}, avoiding reliance on variational calculus and circumventing the potential risks of a lack of rigor that may arise in the use of variational methods, which typically entails more intricate and complex analysis to ensure rigor.

\section{Background and Problem Setup}

\subsection{Fundamentals of Weak Transport Theory}
Optimal transport theory studies the transportation of mass between two 
probability measures. Given two probability measures $\mu$ and $\nu$ on Polish spaces   (complete separable metric spaces) $(\mathcal{X},d_{\X}),(\mathcal{Y},d_{\Y})$, the optimal transport problem is to find a
transport plan $\pi$ that minimizes the transportation cost
\begin{equation}
    \inf_{\pi\in\Pi(\mu,\nu)}\int_{\mathcal{X}\times\mathcal{Y}}c(x,y)d\pi(x,y)
\end{equation}
where $\Pi(\mu,\nu)$ is the set of all transport plans between $\mu$ and $\nu$, and $c(x,y)$ is the cost of transporting mass from $x$ to $y$. The optimal transport problem has been widely studied in mathematics, and has applications in various fields, such as machine learning\cite{courty2016optimal,hamri2021regularized,arjovsky2017wasserstein,chen2024multi}, statistics\cite{talagrand1996transportation,cao2024connecting,gunsilius2021matching}, and information theory\cite{wibisono2018convexity,courtade2019transportation,bai2022optimal,bai2023information}.  Readers are referred to \cite{villani2009optimal} for background on optimal transport.

Weak transport theory was initiated by Gozlan et al.\cite{gozlan2017kantorovich} motivated by applications in geometric inequalities. Their contribution sparked extensive research activity among various groups and led to applications in numerous fields in \cite{gozlan2018characterization,gozlan2020mixture,backhoff2019existence,backhoff2020weak,backhoff2022applications,backhoff2022stability,acciaio2021weak}. Throughout $\X$ and $\Y$ denote Polish spaces. Given probability measures $\mu\in\pc(\X),\nu\in\pc(\Y)$ we write $(\mu,\nu)$ for the set of all couplings on $\X\times\Y$ with marginals $\mu$ and $\nu$. Given a coupling $\pi$ on $\X\times\Y$ , we denote a regular disintegration with respect to the first marginal by $(\pi_x),x\in\X$. We consider cost functionals of the form
\begin{equation}
    C:\X\times\pc(\Y)\to\R\cup\{+\infty\}
\end{equation}
where $C$ is lower bounded and lower semicontinuous, and $C(x,\cdot)$ is assumed to be convex on $\pc(\Y)$ for all $x\in\X$. The weak transport problem is to find a coupling $\pi$ that minimizes the cost functional
\begin{equation}\label{eq:wot}
    V_C(\mu,\nu)=\inf_{\pi\in\Pi(\mu,\nu)}\int_{\X}C(x,\pi_x)d\mu(dx).
\end{equation}
The classical transport problem is included via $C(x,p)=\int c(x,y) dp(y)$ for some cost function $c:\X\times\Y\to\R\cup\{+\infty\}$. For $t\geqslant 1$, $\pc^t_{d_{\Y}}(\Y)$, which, for convenience, can be abbreviated as $\pc_t(\Y)$, denotes the set of Borel probability measures with finite $t$-th moment for some fixed metric $d_\Y$ (compatible with the topology on $\Y$), that is, a Borel probability measure $\nu$ is in $\pc_t(\Y)$ if and only if for some $y_0 \in \Y$ we have
$$\int_{\Y}d_\Y (y,y_0)^t \nu(dy)<\infty.$$
The set of continuous functions on $\Y$ which are dominated by a multiple of $1+d_\Y (y,y_0)^t$, is denoted by $\Phi_t(\Y)$. We equip the set of probability measures $\pc_t(\Y)$ with the $t$-th Wasserstein topology. Specifically, a sequence $(\nu_k)_{k\in\mathbb{N}}$ converges to $\nu\in\pc_t(\Y)$ if $\nu_k (f)\coloneqq \int f d\nu_k$ converges to $\nu(f)$ for all $f\in \Phi_t(\Y)$. The space $\pc(\Y)$ itself is equipped with the usual weak topology. The same conventions apply to $\X$ instead of $\Y$.

\begin{definition}{\bm{$(A)$}}\label{def:(A)}
    A cost functional $C:\X\times\pc_t(\Y)\to\R\cup\{+\infty\}$ satisfies property (A) if and only if
    \begin{enumerate}
        \item $C$ is lower semicontinuous with respect to the product topology on $\X\times\pc_t(\Y)$,
        \item C is bounded from below,
        \item the map $p\mapsto C(x,p)$ is convex, that is, for all $x\in\X$ and all $p,q\in\pc_t(\Y)$, $\lambda\in [0,1]$, we have
        $$C(x,\lambda p+(1-\lambda) q)\leqslant \lambda C(x,p)+(1-\lambda) C(x,q).$$
    \end{enumerate}
\end{definition}

\begin{theorem}\label{thm:existence-1}[Theorem 1.1 in \cite{backhoff2019existence}]
Assume that $C:\X\times\pc(\Y)\to\R\cup\{+\infty\}$ is jointly lower semicontinuous, bounded from below and convex in the second argument. Then the problem
$$\inf_{\pi\in\Pi(\mu,\nu)} \int_{\X}C(x,\pi_x) \mu(dx)$$
admits a minimizer.
\end{theorem}
However, the cost function $C$ may not be lower semicontinuous with respect to weak convergence in many cases, such as $\rho(x,y)=\norm{x-y}^2$ for $\X=\Y=\R^n$, hence we will need to employ a refined version of Existence Theorem \ref{thm:existence-1} to carry out applications in general cases.

\begin{theorem}{\textbf{(Existence and semicontinuity)}}\label{thm:existence}[Theorem 3.2 in \cite{backhoff2022applications}]
   If the cost function $C$ satisfies property $(A)$, the infimum in \eqref{eq:wot} is attained and the value $V_C(\mu,\nu)$ depends in a lower semicontinuous way on the marginals $(\mu,\nu)\in\pc(\X)\times \pc_t(\Y)$.
\end{theorem}


\subsection{A Sketch of The Sch\"odinger Problem}

The (static) Sch\"odinger problem (SP), also called the Sch\"odinger bridge problem, was addressed in \cite{schrodinger1931umkehrung} in 1931. While the problem is historically old, it has  attracted significant attention in the past decades due to its deep connection with various fields, especially (entropic) optimal transport theory. In modern terminology, SP can be expressed as follows. For two Polish spaces $\X,\Y$ with given probability measures $\mu\in\pc(\X),\nu\in\pc(\Y)$, and a given reference probability measure $\gamma\in\pc(\X\times\Y)$. We denote by $\Pi(\mu,\nu)\subset \pc(\X\times\Y)$ the set of couplings; namely, the set of all $\pi\in\pc(\X\times\Y)$ satisfying 
\begin{equation}
    \begin{aligned}
        \int_{\X\times\Y}f(x)\pi(dx,dy)&=\int_{\X}f(x)\mu(dx)\\
        \int_{\X\times\Y}g(y)\pi(dx,dy)&=\int_{\Y}g(y)\nu(dy)
    \end{aligned}
\end{equation}
for all bounded measurable functions $f,g$. The Schr\"odinger problem is to find a coupling $\pi\in\Pi(\mu,\nu)$ that minimizes the relative entropy with respect to the reference measure $\gamma$:
\begin{equation}
    \begin{aligned}
        H(\pi|\gamma)=\begin{cases}
        \int_{\X\times\Y} \log \dfrac{d\pi}{d\gamma}d\pi, & \text{if}\; \pi\ll\gamma\\
        +\infty, & \text{otherwise}
        \end{cases}
    \end{aligned}
\end{equation}
This problem can be seen as a special case of OWT problems. Assume that the value $\inf\limits_{\pi\in\Pi(\mu,\nu)}H(\pi|\gamma)$ is finite, then OWT guarantees the existence of minimizer and provides the structure of optimizers in the SP. 
\begin{lemma}\label{lemma: explicit expression}
    [Corollary 4.3 in \cite{backhoff2022applications}]
Let $\mu\in\pc(\X)$, $\nu\in\pc(\Y)$ and $\gamma\in\pc(\X\times\Y)$ a probability measure equivalent to $\mu\times\nu$. Assume that the value of the corresponding entropic optimal transport problem is finite, that is
\begin{equation}
    \inf_{\pi\in\Pi(\mu,\nu)}H(\pi|\gamma)<+\infty.
\end{equation}
A coupling $\pi^\star$ minimizes the above problem if and only if $H(\pi^\star|\gamma)<\infty$ and there exist measurable functions $f,g$ such that 
$$\dfrac{d\pi^\star}{d\gamma}(x,y)=f(x)g(y), \gamma-a.e.$$
\end{lemma}
In fact, $f(x)$ and $g(y)$ in the lemma above are nonnegative and unique up to a multiplicative constant\cite{nutz2021introduction}. In the case where $d\gamma=K e^{-\beta\rho(x,y)}d\mu\times d\nu$ for a nonnegative measurable function $\rho(x,y)$, 
since $\pi^\star\in \Pi(\mu,\nu)$, we have $\int_{\Y}d\pi^\star=d\mu, \int_{\X}d\pi^\star=d\nu$, which can be written as Schr\"odinger equations\cite{nutz2021introduction}
\begin{equation}\label{Schr\"odinger equations}
    \begin{aligned}
       1&=Kf(x)\int_{\Y}g(y)e^{-\beta\rho(x,y)}d\nu\\ 
       &=Kg(y)\int_{\X}f(x)e^{-\beta\rho(x,y)}d\mu \\ 
       &=\int_{\X}\dfrac{g(y')e^{-\beta\rho(x,y')}}{\int_{\Y}g(y)e^{-\beta\rho(x,y)}d\nu}d\mu,\quad\text{$\forall y'\in\Y$}.
    \end{aligned}
\end{equation}
where
\begin{equation}
    \begin{aligned}
        K&=\left(\int_{\X\times\Y}e^{-\beta\rho(x,y)}d\mu d\nu\right)^{-1}\\ 
        &=\left(\int_{\X\times\Y}f(x)g(y)e^{-\beta\rho(x,y)}d\mu d\nu\right)^{-1}
    \end{aligned}
\end{equation}
In the following sections, we will recapitulate the RD theory, and investigate its connection to the Schr\"odinger problem within the framework of weak transport theory.

\subsection{A Recap of Rate-distortion Theory}
Given a source $X$ with probability distribution $p(x)$, and a distortion measure $\rho(x,y)$, the RD function $R(D)$ is defined as the minimum rate required to achieve a distortion level $D$,

\begin{equation}
    R(D)=\inf_{Y}I(X;Y)
\end{equation}

\noindent subject to the constraint

\begin{equation}
    \E[\rho(X,Y)]\leqslant D
\end{equation}
    where $I(X;Y)$ is the mutual information between $X$ and $Y$, and $Y$ is the reconstruction. I. Csisz\'ar proved the following basic properties about the rate-distortion function and its parametric representation under sufficiently weak assumptions in \cite{csiszar1974extremum}, providing a solid theoretical foundation for the RD theory:
    
    \begin{assumption}\label{assumption: distortion function basic requirement}
        Let $\rho:\X\times\Y\to [0,\infty]$ be a given measurable function, called a loss function, with
        \begin{equation}
            \inf_{y} \rho(x,y)=0, \forall x\in\X.
        \end{equation}
    \end{assumption}
    \begin{lemma}{}
        [Lemma 1.1 in \cite{csiszar1974extremum}]Under Assumption \ref{assumption: distortion function basic requirement}, the rate-distortion function $R(D)$ is a non-increasing convex function of $D$ and one may write
    \begin{equation}
        R(D)=\inf_{\E_{\pi}(\rho)=D} I(X;Y)
    \end{equation}
    where $D\leqslant D_{1}$, which is the smallest non-negative number such that $R(D)$ is constant in $(D_1,+\infty)$.
    \end{lemma}
    \begin{assumption}\label{assumption: existence of random variables}
        We assume the existence of random variables $\xi,\eta$ with $\xi \sim \mu$ such that
        \begin{equation}
            I(\xi;\eta)<\infty, \E\rho(\xi,\eta)<\infty.
        \end{equation}
    \end{assumption}
    This assumption can exclude the contingency that $R(D)$ is identically infinite.
    \begin{assumption}\label{assumption: existence of finite set}
       Suppose that there exists a finite set $B\subset \Y$ such that
       $$\int \rho(x,B) \mu(dx)<\infty\quad \text{where}\quad \rho(x,B)\coloneqq \min_{y\in B}\rho(x,y).$$
    \end{assumption}
    \noindent Define $D_{\min}\coloneqq \inf\{D\geqslant 0|R(D)<\infty\}$.
    \begin{theorem}\label{thm: R(D) finite}
        [Theorem 2.1 in \cite{csiszar1974extremum}]
        Under Assumption \ref{assumption: distortion function basic requirement} and \ref{assumption: existence of random variables}, suppose that to any $\epsilon>0$ there exists measurable subsets $A_1\subset A_2\subset \cdots$ of $\X$ and finite subsets $B_1\subset B_2\subset \cdots$ of $\Y$ such that $\mu(A_n)\to 1$ as $n\to\infty$ and 
        \begin{equation}\label{eq:assertion 1}
            \min_{y\in B_n}\rho(x,y)<\epsilon, \forall x\in A_n, n=1,2,\cdots
        \end{equation}
        Then $R(D)$ is finite for all $D>0$, namely, $D_{\min}=0$.
    \end{theorem}
    The hypothesis of this theorem is trivially fulfilled if $\Y$ is at most countable. It is also fulfilled if $\X$ is Polish and $\rho^{(K)}(x,y)\coloneqq \min\{\rho(x,y),K\}$ is upper semicontinuous for some fixed $K>0$. In this case the $A_i$'s may be chosen as an increasing sequence of compact subsets of $\X$, since on a Polish space every probability distribution is tight\cite{csiszar1974extremum}.
   
    \begin{theorem}{}
        [Theorem 2.1 in \cite{csiszar1974extremum}]
        Under Assumption \ref{assumption: distortion function basic requirement} and \ref{assumption: existence of random variables}, if there exist $A_n$'s as in Theorem \ref{thm: R(D) finite} and $y_n\in \Y$ with 
        \begin{equation}
            \max_{x\in A_n} \rho(x,y_n)<\infty, n=1,2,\cdots
        \end{equation}
        then $R(D)\to 0$ as $D\to\infty$. Furthermore, setting
        \begin{equation}
            D_{\max}=\lim_{K\to\infty}\left(
                \inf_{y}\int \rho^{(K)}(x,y) \mu(dx)
            \right).
        \end{equation}
        Then $R(D)$ is positive for $D<D_{\max}$ and if $D_{\max}<\infty$ then $R(D)=0$ for $D>D_{\max}$.
    \end{theorem}

\begin{theorem}{}
    [Theorem 2.3 in \cite{csiszar1974extremum}]
  Under Assumption \ref{assumption: distortion function basic requirement}, \ref{assumption: existence of random variables} and \ref{assumption: existence of finite set}, for each $D>D_{\min}$ holds
  \begin{equation}
    R(D)=\max_{\beta,\alpha(x)}\left\{
        \int \log \alpha(x) \mu(dx) -\beta D
    \right\}
  \end{equation}
  where the maximization refers to $\alpha(x)\geqslant 1$ and $\beta\geqslant 0$ under the constraints
  \begin{equation}
    \int \alpha(x)e^{-\beta\rho(x,y)} \mu(dx)\leqslant 1, \forall y\in\Y.
  \end{equation}
  The maximizing $\beta$ values are exact those associated with $D$ by
  \begin{equation}
    R(D')+\beta D'\geqslant R(D)+\beta D, \forall D'.
  \end{equation}
\end{theorem}

\section{OWT Meets Rate-distortion Theory}

In this section, we will revisit rate-distortion theory using the framework of weak transport theory. We begin by transforming the original RD problem, which is focused on optimizing the disintegration, into an OWT problem of optimizing the joint distribution under the condition of a given target measure (or reconstruction). This is followed by seeking the infimum over the target measure. We will utilize the existence theorem for weak transport problems, along with its application in the Schr\"odinger problem, to provide a new parametric representation of the RD function. Furthermore, under the condition of the existence of an optimal reconstruction, we will prove the necessary optimality conditions for the RD function, thereby offering a concise proof of the SLB achievability as proposed by K. Rose\cite{rose1994mapping}. Our main contributions are summarized from Theorem \ref{thm: refined existence} to Theorem \ref{thm: parametric representation-2}. We provide detailed proofs of them in Appendix.

\begin{lemma}\label{lemma: transform the rate-distortion problem}
    Under Assumption \ref{assumption: distortion function basic requirement} and \ref{assumption: existence of random variables}, the rate-distortion function $R(D)$ can be represented as
    \begin{equation*}
        \begin{aligned}
            \inf_{\pi_x:\E_\pi \rho \leqslant D} I(X;Y)&=\inf_{\nu\in\pc(\Y)}\inf_{\pi\in\Pi(\mu,\nu): \E_\pi \rho \leqslant D}I(X;Y)\\ 
            &=\inf_{(\nu,\pi)\in \pc(\Y)\times \Pi(\mu,\nu):\E_\pi\rho\leqslant D}I(X;Y).
        \end{aligned}
    \end{equation*}
\end{lemma}
\begin{proof}
See Appendix \ref{appendix: proof of lemma: transform the rate-distortion problem}.
\end{proof}
\begin{remark}
    Note that $I(X;Y)=\int_{\X\times\Y}\log\dfrac{d\pi}{d\mu d\nu} d\pi=\int_{\X}d\mu(x)\int_{\Y}\log\dfrac{d\pi_x}{d\nu}d\pi_x$, and $\pi\ll\mu\times\nu\Longleftrightarrow \pi_x\ll\nu$. Hence we always assume that $\pi_x\ll \nu$ so that the infimum is proper.
\end{remark}

\noindent Consider the (weak) dual problem of rate-distortion problem
\begin{equation}
    \begin{aligned}
    R(D)&=\inf_{\nu\in\pc(\Y)}\inf_{\pi\in\Pi(\mu,\nu): \E_\pi \rho \leqslant D}I(X;Y)\\
    &\geqslant \inf_{\nu\in\pc(\Y)}\sup_{\beta\geqslant 0}\inf_{\pi\in\Pi(\mu,\nu)}I(X;Y)+\beta (\E_\pi \rho -D).
    \end{aligned}
\end{equation}
Define 
\begin{equation}
    \begin{aligned}
        J(\nu,\beta)&\coloneqq \inf_{\pi\in\Pi(\mu,\nu)} I(X;Y)+\beta (\E_\pi\rho -D)\\ 
        &=\int_{\X}d\mu(x) \int_{\Y} \left[ 
            \log\dfrac{d\pi_x}{d\nu} +\beta\rho(x,y)-\beta D
        \right] d\pi_x\\ 
        d\gamma(x,y)&\coloneqq K e^{-\beta \rho(x,y)} d\mu(x)d\nu(y)
    \end{aligned}
\end{equation}
where $K=\left( 
    \int_{\X\times\Y} e^{-\beta\rho(x,y)} d\mu d\nu
\right)^{-1}$ such that $\gamma\ll \mu\times\nu$ is a probability measure on $\X\times \Y$   (note that $\gamma$ is indeed equivalent to $\mu\times\nu$).
We have
\begin{equation*}
	\begin{aligned}
        &I(X;Y)+\beta (\E_\pi\rho -D)\\
        &\;=\int_{\X} D_{KL}(\pi_x\|\gamma_x)d\mu+\int_{\X}\log K_x d\mu(x)-\beta D\\
        &\;=\int_{\X\times\Y}\log\dfrac{d\pi_x}{d\gamma_x}d\pi-\int_{\X}\log\left(\int_{\Y}e^{-\beta\rho(x,y)}d\nu\right) d\mu(x)-\beta D\\
        &\;=\int_{\X\times\Y}\log\dfrac{d\pi}{d\gamma}d\pi-\int_{\X}\log\left(\int_{\Y}e^{-\beta\rho(x,y)}d\nu\right) d\mu(x)\\ 
        &\qquad -\beta D+\int_{\X}\log\dfrac{d\gamma_0}{d\mu}d\mu
	\end{aligned}
\end{equation*}
where $\gamma_0$ is the $X$-marginal of $\gamma$ and 
\begin{equation*}
	\begin{aligned}
       d\gamma_x&=K_x e^{-\beta \rho(x,y)}d\nu, K_x=\left(\int_{\Y}e^{-\beta\rho(x,y)}d\nu\right)^{-1}.
	\end{aligned}
\end{equation*}
Hence we only need to consider the infimum of $D_{KL}(\pi\| \gamma)$ over $\pi\in\Pi(\mu,\nu)$.

\begin{remark}
   
    One might be concerned about the finiteness of $K,K_x$ and integrability of $\log K_x$, while they always hold indeed. It is noticeable that $K=+\infty$ if and only if $\rho(x,y)=+\infty, \mu\times \nu$-a.e., implying that $\rho=+\infty, \pi$-a.e. for each $\pi\in\Pi(\mu,\nu)$ with $\pi_x\ll \nu$. In this case, $\E_\pi \rho = +\infty$ and $J(\nu,\beta)=+\infty$. Hence we only need to consider the case that $K<+\infty$. Similarly, $K_x=+\infty$ if and only if $\rho(x,y)=+\infty, \nu$-a.e. for given $x\in\X$, implying that $\rho=+\infty, \pi_x$-a.e. with $\pi_x\ll \nu$. Then we have $K_x<+\infty, \mu$-a.e. otherwise $J(\nu,\beta)$ becomes infinite. Hence the case of $K_x=+\infty$ can be ignored. Moreover, if $\int_{\X}\log K_x d\mu(x)=+\infty$, we also have $J(\nu,\beta)=+\infty$, which can be excluded.
\end{remark}
We consider the existence of the minimizer of $J(\nu,\beta)$. Define
\begin{equation}
    \begin{aligned}
    C(x,p)&\coloneqq \int_{\Y} \left[ 
        \log\dfrac{dp}{d\nu}+\beta\rho(x,y)-\beta D
    \right] dp, p\in\pc(\Y).
    \end{aligned}
\end{equation}
\begin{assumption}\label{assumption: lower semicontinuity of C}
    $\int_{\Y} \rho(x,y) dp$ is jointly lower semicontinuous in $(x,p)\in\X\times\pc(\Y)$.
\end{assumption}
This assumption can hold if $\X,\Y$ are Polish spaces and $\rho\in C_b(\X\times \Y)$. It can be checked that $(x,p)\to \int_{\Y} \rho(x,y) dp$ is indeed bounded and jointly continuous.  In addition, if $\rho(x,y)$ is jointly lower semicontinuous with an upper bound (the lower bound is $0$), then it can be increasingly converged by a sequence of bounded $k$-Lipschitz functions, which implies that $C(x,p)$ is jointly lower semicontinuous with respect to the product topology on $\X\times\pc(\Y)$. One way to verify these assertions is to use \textit{Theorem A.3.12} in \cite{dupuis2011weak}.

\begin{lemma}\label{lemma: lower semicontinuity of C}
    Under Assumption \ref{assumption: lower semicontinuity of C}, the cost function $C(x,p)$ is jointly lower semicontinuous, bounded from below and convex in the second argument.
\end{lemma}
\begin{proof}
See Appendix \ref{appendix: proof of lemma: lower semicontinuity of C}.
\end{proof}
\begin{corollary}
    Under the Assumption \ref{assumption: lower semicontinuity of C}, $J(\nu,\beta)$ admits a minimizer, that is, there exists $\pi^\star\in \Pi(\mu,\nu)$ such that 
    \begin{equation*}
        \begin{aligned}
        J&(\nu,\beta)=\int_{\X}d\mu(x) \int_{\Y} \left[ 
            \log\dfrac{d\pi^\star_x}{d\nu} +\beta\rho(x,y)-\beta D
        \right] d\pi^\star_x\\ 
        &=\int_{\X\times\Y}\log\dfrac{d\pi^\star}{d\gamma}d\pi^\star-\int_{\X}\log\left(\int_{\Y}e^{-\beta\rho(x,y)}d\nu\right) d\mu(x)\\
        &\qquad -\beta D+\int_{\X}\log\dfrac{d\gamma_0}{d\mu}d\mu.
        \end{aligned}
    \end{equation*}
\end{corollary}
However, Assumption \ref{assumption: lower semicontinuity of C} may not hold in many cases, such as $\rho(x,y)=\norm{x-y}^2$ for $\X=\Y=\R^n$, which has been mentioned below Theorem \ref{thm:existence-1}. Hence we need to utilize the Refined Existence heorem \ref{thm:existence} to carry out applications in general cases. Before this, we need to impose some addtional conditions on $\rho(x,y)$ and the source $\mu$.
\begin{assumption}\label{assumption: refined lower semicontinuity of C}
$\int_{\Y}\rho(x,y)dp$ is jointly lower semicontinuous in $(x,p)\in\X\times\pc_t(\Y)$ for given metric $d_{\Y}$ on space $\Y$.
\end{assumption}

This assumption is weaker than Assumption \ref{assumption: lower semicontinuity of C}, as convergence (or topology) in the space $\pc_t(\Y)$ is much stronger than in $\pc(\Y)$. This extension broadens the applicable scenarios, such as the case where the cost function takes the form $\rho(x,y) = \|x - y\|^t$ with \( t \geqslant 1 \) and \(\X = \Y = \mathbb{R}^n\) under the Euclidean metric. We can also demonstrate that the cost function $C(x,p)$ satisfies property $(A)$.
\begin{theorem}\label{thm: refined existence}
    Under Assumption \ref{assumption: refined lower semicontinuity of C}, the cost function $C(x,p)$ satisfies property $(A)$. Assume that $\nu\in\pc_t(\Y)$. Then $J(\nu,\beta)$ admits a minimizer $\pi^{\star}$.
\end{theorem}
\begin{proof}
See Appendix \ref{appendix: proof of refined existence}.
\end{proof}
In order to restrict the reconstruction to the space $\pc_t(\Y)$, we introduce the following assumption, which allows us to apply the Refined Existence Theorem \ref{thm:existence} of OWT to guarantee the existence of $\pi^{\star}$. Although this assumption is stronger than those in the works of I. Csisz\'ar\cite{csiszar1974extremum} and F. Rezaei\cite{rezaei2006rate} et al., it remains reasonable and widely applicable in practical scenarios.
\begin{assumption}\label{assumption: retrict the target measure}
    $\X=\Y$ are Polish spaces and the source $\mu\in\pc_t(\X)$. The loss function $\rho(x,y)$ satisfies that 
    $$\rho(x,y)\geqslant c \cdot d_{\X}(x,y)^t\;\;\text{for some constant $c>0$.}$$
\end{assumption}
\begin{proposition}\label{proposition: the target measure has finite t-Wasserstein distance}
The reconstruction distribution $\nu\in\pc_t(\Y)$ for any $D>D_{\min}$ under Assumption \ref{assumption: distortion function basic requirement}, \ref{assumption: existence of random variables} and \ref{assumption: retrict the target measure}.
\end{proposition}
\begin{proof}
See Appendix \ref{appendix: proof of proposition 1}.
\end{proof}
\noindent By Proposition \ref{proposition: the target measure has finite t-Wasserstein distance} and Lemma \ref{lemma: transform the rate-distortion problem}, we can obtain that
\begin{equation}
    \begin{aligned}
        R(D)&=\inf_{\pi_x:\E_\pi \rho \leqslant D}I(X;Y)\\ 
        &=\inf_{\nu\in\pc_t(\Y)}\inf_{\pi\in\Pi(\mu,\nu): \E_\pi \rho \leqslant D}I(X;Y)\\ 
        &\geqslant  \inf_{\nu\in\pc_t(\Y)}\sup_{\beta\geqslant 0}\inf_{\pi\in\Pi(\mu,\nu)}I(X;Y)+\beta (\E_\pi \rho -D)\\ 
        &=\inf_{\nu\in\pc_t(\Y)}\sup_{\beta\geqslant 0} J(\nu,\beta).
    \end{aligned}
\end{equation}
In fact, by applying the existence theorem of OWT and relevant results from the Schr\"odinger problem\cite{leonard2013survey,nutz2021introduction}, we can directly obtain the expression for $J(\nu, \beta)$, which in turn provides a new parametric representation for $R(D)$.
 By Lemma \ref{lemma: explicit expression}, we can obtain the parametric representation of $R(D)$ and a series of equivalent forms of the RD function.
\begin{theorem}\label{thm: parametric representation-1}
Under Assumption \ref{assumption: distortion function basic requirement}$-$\ref{assumption: retrict the target measure}, the RD function $R(D)$ can be represented as

\begin{equation*}
	\begin{aligned}
&R(D)=\inf\limits_{\pi_x:\mathbb{E}_{\pi}\rho\leqslant D}I(X;Y)\\
&=\inf\limits_{\nu\in\pc_t(\Y)}\sup_{\beta\geqslant 0}\inf\limits_{\pi\in\Pi(\mu,\nu)}I(X;Y)+\beta(\mathbb{E}_{\pi}\rho-D)\\
&=\inf\limits_{\nu}\max_{\beta\in \partial R(D)}
\left\{
	\!-\!\int_{}\log\left(\int_{}e^{-\beta\rho(x,y)}d\nu\right)d\mu\!-\!\beta D\!+\!L(\nu,\beta)
\right\}\\ 
&=
\max_{\beta\in \partial R(D)}\inf\limits_{\nu}\!
\left\{
	\!-\!\int_{}\log\left(\int_{}e^{-\beta\rho(x,y)}d\nu\right)d\mu\!-\!\beta D\!+\!L(\nu,\beta)
\right\}\\ 
&=\!\inf\limits_{\nu}\!
\left\{\!
	-\!\int\!\log\!\left(\!\int \!e^{\!-\beta\rho(x,y)}d\nu\!\right)\!d\mu\!-\!\beta D\!+\!L(\nu,\beta)\!
\right\}\!,\forall \beta\!\in\!\partial R(D)
	\end{aligned}
  \end{equation*}
where 
$$L(\nu,\beta)=\int_{\X}\log\left(\dfrac{\int_{\Y}e^{-\beta\rho(x,y)}d\nu}{\int_{\Y}g(y)e^{-\beta\rho(x,y)}d\nu}\right)d\mu +\int_{\Y}\log g(y)d\nu.$$
with $g(y)$ satisfying the Schr\"odinger equations \eqref{Schr\"odinger equations}
\end{theorem}
\begin{proof}
See Appendix \ref{appendix: proof of thm: parametric representation-1}.
\end{proof}
If we further assume the existence of the optimal reconstruction $\nu^\star$, $L(\nu,\beta)$ equals $0$ so that the parametric representation \eqref{simplified pr} can be reproduced.
\begin{theorem}\label{thm: parametric representation-2}
    Under the set of Theorem \ref{thm: parametric representation-1}, if additionally, there exists an optimal reconstruction $\nu^\star$, then the corresponding optimal joint distribution satisfies that
    \begin{equation}\label{eq: optimal joint distribution}
    d\pi^\star = \dfrac{e^{-\beta\rho(x,y)}}{\int_{\Y}e^{-\beta\rho(x,y)}d\nu^\star}d\mu\times d\nu^\star
    \end{equation}
    which is equivalent to that $g(y)$ in Lemma \ref{lemma: explicit expression} is a constant, namely, $L(\nu^{\star},\beta)=0$ in Theorem \ref{thm: parametric representation-1}. Then the RD function $R(D)$ can be represented as
    \begin{equation}\label{simplified pr}
        \begin{aligned}
        R(D)&=\inf\limits_{\nu}
        \left\{
            -\int_{\X}\log\left(\int_{\Y}e^{-\beta\rho(x,y)}d\nu\right)d\mu-\beta D
        \right\}\\ 
        &=-\int_{\X}\log\left(\int_{\Y}e^{-\beta\rho(x,y)}d\nu^{\star}\right)d\mu-\beta D
        \end{aligned}
    \end{equation}
    holding for any $\beta\in\partial R(D)$ and $\nu^{\star}$ does not change even though $\beta$ can vary in $\partial R(D)$.
\end{theorem}
\begin{proof}
See Appendix \ref{appendix: proof of thm: parametric representation-2}.
\end{proof}
By \eqref{eq: optimal joint distribution} along with K. Rose's methods of using the completeness of Hermite polynomials, we can reproduce the following conclusion.

\begin{corollary}\label{corollary: K.Rose's conclusion about the achievability of SLB}[Theorem 2 and Corollary 2 in \cite{rose1994mapping}]
For $\X=\Y=\R^n$ and $\rho(x,y)=\norm{x-y}^2$, if the support of the optimal reproduction random variable has an accumulation point, then the RD function coincides with the SLB. 

If the SLB does not hold with equality, then the support of the optimal reproduction random variable consists of isolated singularities. Further, if this support is bounded, then Y is discrete and finite.
\end{corollary}

\section{Discussion and Future Work}
In this paper, we revisit RD theory from the perspective of OWT. Within the framework of OWT, we derive a parametric representation of the RD function based on the existence theorem of OWT, thereby linking the RD function to the Schr\"odinger bridge problem. Building upon this, we provide a concise proof of the optimality condition \eqref{eq: optimal joint distribution} under the existence of an optimal reconstruction. Potential directions for future research include relaxing the current assumptions and exploring novel numerical algorithms to solve the RD problem. This can be inspired by the numerical techniques developed for solving the Schrödinger bridge problem, as discussed in \cite{pooladian2024plug,caluya2021wasserstein,marino2020optimal}.  

\bibliographystyle{IEEEtran}
\bibliography{references}

\begin{thebibliography}{10}
\providecommand{\url}[1]{#1}
\csname url@samestyle\endcsname
\providecommand{\newblock}{\relax}
\providecommand{\bibinfo}[2]{#2}
\providecommand{\BIBentrySTDinterwordspacing}{\spaceskip=0pt\relax}
\providecommand{\BIBentryALTinterwordstretchfactor}{4}
\providecommand{\BIBentryALTinterwordspacing}{\spaceskip=\fontdimen2\font plus
\BIBentryALTinterwordstretchfactor\fontdimen3\font minus
  \fontdimen4\font\relax}
\providecommand{\BIBforeignlanguage}[2]{{%
\expandafter\ifx\csname l@#1\endcsname\relax
\typeout{** WARNING: IEEEtran.bst: No hyphenation pattern has been}%
\typeout{** loaded for the language `#1'. Using the pattern for}%
\typeout{** the default language instead.}%
\else
\language=\csname l@#1\endcsname
\fi
#2}}
\providecommand{\BIBdecl}{\relax}
\BIBdecl

\bibitem{shannon1948mathematical}
C.~E. Shannon, ``A mathematical theory of communication,'' \emph{The Bell
  System Technical Journal}, vol.~27, no.~3, pp. 379--423, 1948.

\bibitem{shannon1959coding}
------, ``Coding theorems for a discrete source with a fidelity criterion,''
  \emph{International Convention Record}, vol.~7, pp. 325--350, 1959.

\bibitem{berger2003rate}
T.~Berger, \emph{Rate-distortion theory: A mathematical basis for data
  compression}.\hskip 1em plus 0.5em minus 0.4em\relax Englewood Cliffs:
  Prentice-Hall, 1971.

\bibitem{blahut1972computation}
R.~Blahut, ``Computation of channel capacity and rate-distortion functions,''
  \emph{IEEE Trans. Inf. Theory}, vol.~18, no.~4, pp. 460--473, 1972.

\bibitem{blahut1987principles}
R.~E. Blahut, \emph{Principles and practice of information theory}.\hskip 1em
  plus 0.5em minus 0.4em\relax Addison-Wesley Longman Publishing Co., Inc.,
  1987.

\bibitem{cover2006elements}
T.~M. Cover and J.~A. Thomas, \emph{Elements of Information Theory},
  2nd~ed.\hskip 1em plus 0.5em minus 0.4em\relax New York, NY, USA: Wiley,
  2006.

\bibitem{csiszar1974extremum}
I.~Csisz{\'a}r, ``On an extremum problem of information theory,'' \emph{Studia
  Scientiarum Mathematicarum Hungarica}, vol.~9, no.~1, pp. 57--71, 1974.

\bibitem{arimoto1972algorithm}
S.~Arimoto, ``An algorithm for computing the capacity of arbitrary discrete
  memoryless channels,'' \emph{IEEE Trans. Inf. Theory}, vol.~18, no.~1, pp.
  14--20, 1972.

\bibitem{rezaei2006rate}
F.~Rezaei, N.~Ahmed, and C.~D. Charalambous, ``Rate distortion theory for
  general sources with potential application to image compression,''
  \emph{International Journal of Applied Mathematical Sciences}, vol.~3, no.~2,
  pp. 141--165, 2006.

\bibitem{riegler2018rate}
E.~Riegler, H.~B{\"o}lcskei, and G.~Koliander, ``Rate-distortion theory for
  general sets and measures,'' in \emph{Proc. lEEE Int. Symp. Inf. Theory
  (ISIT)}, Vail, CO, USA, 2018, pp. 101--105.

\bibitem{riegler2023lossy}
E.~Riegler, G.~Koliander, and H.~B{\"o}lcskei, ``Lossy compression of general
  random variables,'' \emph{Information and Inference: A Journal of the IMA},
  vol.~12, no.~3, pp. 1759--1829, 2023.

\bibitem{kawabata1994rate}
T.~Kawabata and A.~Dembo, ``The rate-distortion dimension of sets and
  measures,'' \emph{IEEE Trans. Inf. Theory}, vol.~40, no.~5, pp. 1564--1572,
  1994.

\bibitem{kostina2019successive}
V.~Kostina and E.~Tuncel, ``Successive refinement of abstract sources,''
  \emph{IEEE Trans. Inf. Theory}, vol.~65, no.~10, pp. 6385--6398, 2019.

\bibitem{donsker1975asymptotic}
M.~D. Donsker and S.~S. Varadhan, ``Asymptotic evaluation of certain {M}arkov
  process expectations for large time, {I},'' \emph{Commun. Pure Appl. Math.},
  vol.~28, no.~1, pp. 1--47, 1975.

\bibitem{equitz1991successive}
W.~H. Equitz and T.~M. Cover, ``Successive refinement of information,''
  \emph{IEEE Trans. Inf. Theory}, vol.~37, no.~2, pp. 269--275, 1991.

\bibitem{gray2013transportation}
R.~M. Gray, ``Transportation distance, shannon information, and source
  coding,'' in \emph{GRETSI Symposium on Signal and Image Processing}, 2013,
  https://ee.stanford.edu/~gray/gretsi.pdf.

\bibitem{yang2024estimating}
Y.~Yang, S.~Eckstein, M.~Nutz, and S.~Mandt, ``Estimating the rate-distortion
  function by {W}asserstein gradient descent,'' \emph{Advances in Neural
  Information Processing Systems}, vol.~36, 2024.

\bibitem{gozlan2017kantorovich}
N.~Gozlan, C.~Roberto, P.-M. Samson, and P.~Tetali, ``Kantorovich duality for
  general transport costs and applications,'' \emph{J. Funct. Anal.}, vol. 273,
  no.~11, pp. 3327--3405, 2017.

\bibitem{leonard2013survey}
C.~L{\'e}onard, ``A survey of the {S}chr\"odinger problem and some of its
  connections with optimal transport,'' \emph{arXiv:1308.0215}, 2013.

\bibitem{backhoff2019existence}
J.~Backhoff-Veraguas, M.~Beiglb{\"o}ck, and G.~Pammer, ``Existence, duality,
  and cyclical monotonicity for weak transport costs,'' \emph{Calculus of
  Variations and Partial Differential Equations}, vol.~58, no. 203, 2019.

\bibitem{backhoff2022applications}
J.~Backhoff-Veraguas and G.~Pammer, ``Applications of weak transport theory,''
  \emph{Bernoulli}, vol.~28, no.~1, pp. 370--394, 2022.

\bibitem{rose1994mapping}
K.~Rose, ``A mapping approach to rate-distortion computation and analysis,''
  \emph{IEEE Trans. Inf. Theory}, vol.~40, no.~6, pp. 1939--1952, 1994.

\bibitem{courty2016optimal}
N.~Courty, R.~Flamary, D.~Tuia, and A.~Rakotomamonjy, ``Optimal transport for
  domain adaptation,'' \emph{IEEE Trans. Pattern Analysis and Machine
  Intelligence}, vol.~39, no.~9, pp. 1853--1865, 2016.

\bibitem{hamri2021regularized}
M.~E. Hamri and Y.~Bennani, ``Regularized optimal transport for dynamic
  semi-supervised learning,'' \emph{arXiv:2103.11937}, 2021.

\bibitem{arjovsky2017wasserstein}
M.~Arjovsky, S.~Chintala, and L.~Bottou, ``{W}asserstein generative adversarial
  networks,'' in \emph{International Conference on Machine Learning}.\hskip 1em
  plus 0.5em minus 0.4em\relax PMLR, 2017, pp. 214--223.

\bibitem{chen2024multi}
L.~Chen, Y.~Song, Y.~Cai, J.~Lu, Y.~Li, Y.~Xie, C.~Wang, and G.~He,
  ``Multi-prototype space learning for commonsense-based scene graph
  generation,'' in \emph{Proceedings of the AAAI Conference on Artificial
  Intelligence}, vol.~38, no.~2, 2024, pp. 1129--1137.

\bibitem{talagrand1996transportation}
M.~Talagrand, ``Transportation cost for {G}aussian and other product
  measures,'' \emph{Geometric \& Functional Analysis (GAFA)}, vol.~6, no.~3,
  pp. 587--600, 1996.

\bibitem{cao2024connecting}
H.~Cao, X.~Guo, and M.~Lauri{\`e}re, ``Connecting {GAN}s, mean-field games, and
  optimal transport,'' \emph{SIAM Journal on Applied Mathematics}, vol.~84,
  no.~4, pp. 1255--1287, 2024.

\bibitem{gunsilius2021matching}
F.~Gunsilius and Y.~Xu, ``Matching for causal effects via multimarginal
  unbalanced optimal transport,'' \emph{arXiv:2112.04398}, 2021.

\bibitem{wibisono2018convexity}
A.~Wibisono and V.~Jog, ``Convexity of mutual information along the
  {O}rnstein-{U}hlenbeck flow,'' in \emph{International Symposium on
  Information Theory and Its Applications (ISITA)}, Singapore, 2018, pp.
  55--59.

\bibitem{courtade2019transportation}
T.~A. Courtade, ``Transportation proof of an inequality by {A}nantharam, {J}og
  and {N}air,'' \emph{arXiv:1901.10893}, 2019.

\bibitem{bai2022optimal}
Y.~Bai, ``Optimal transport meets information science: from measure
  concentration, to information theory, to machine learning,'' PhD Thesis,
  University of Delaware, 2022.

\bibitem{bai2023information}
Y.~Bai, X.~Wu, and A.~{\"O}zg{\"u}r, ``Information constrained optimal
  transport: From {T}alagrand, to {M}arton, to {C}over,'' \emph{IEEE Trans.
  Inf. Theory}, vol.~69, no.~4, pp. 2059--2073, 2023.

\bibitem{villani2009optimal}
C.~Villani, \emph{Optimal transport: old and new}, ser. Grundlehren der
  mathematischen Wissenschaften.\hskip 1em plus 0.5em minus 0.4em\relax
  Springer, 2009.

\bibitem{gozlan2018characterization}
N.~Gozlan, C.~Roberto, P.-M. Samson, Y.~Shu, and P.~Tetali, ``Characterization
  of a class of weak transport-entropy inequalities on the line,''
  \emph{Annales de l'Institut Henri Poincaré, Probabilités et Statistiques},
  vol.~54, no.~3, pp. 1667 -- 1693, 2018.

\bibitem{gozlan2020mixture}
N.~Gozlan and N.~Juillet, ``On a mixture of {B}renier and {S}trassen
  theorems,'' \emph{Proc. Lond. Math. Soc.}, vol. 120, no.~3, pp. 434--463,
  2020.

\bibitem{backhoff2020weak}
J.~Backhoff-Veraguas, M.~Beiglb{\"o}ck, and G.~Pammer, ``Weak monotone
  rearrangement on the line,'' \emph{Electronic Communications in Probability},
  vol.~25, pp. 1--16, 2020.

\bibitem{backhoff2022stability}
J.~Backhoff-Veraguas and G.~Pammer, ``Stability of martingale optimal transport
  and weak optimal transport,'' \emph{Ann. Appl. Probab.}, vol.~32, no.~1, pp.
  721--752, 2022.

\bibitem{acciaio2021weak}
B.~Acciaio, M.~Beiglb{\"o}ck, and G.~Pammer, ``Weak transport for non-convex
  costs and model-independence in a fixed-income market,'' \emph{Math.
  Finance}, vol.~31, no.~4, pp. 1423--1453, 2021.

\bibitem{schrodinger1931umkehrung}
E.~Schr{\"o}dinger, ``{\"U}ber die umkehrung der naturgesetze,''
  \emph{Sitzungsberichte der Preussischen Akademie der Wissenschaften.
  Physikalisch-Mathematische Klasse}, vol. 144, pp. 144--153, 1931.

\bibitem{nutz2021introduction}
M.~Nutz, ``Introduction to entropic optimal transport,'' \emph{Lecture notes,
  Columbia University}, 2021.

\bibitem{zou2025revisitratedistortionproblemsoptimal}
J.~Zou, L.~Fan, J.~Gao, and J.~Wang, ``A revisit to rate-distortion theory via
  optimal weak transport,'' \emph{arXiv:2501.09362}, 2025.

\bibitem{dupuis2011weak}
P.~Dupuis and R.~S. Ellis, \emph{A weak convergence approach to the theory of
  large deviations}.\hskip 1em plus 0.5em minus 0.4em\relax John Wiley \& Sons,
  2011.

\bibitem{pooladian2024plug}
A.-A. Pooladian and J.~Niles-Weed, ``Plug-in estimation of {S}chr\"odinger
  bridges,'' \emph{arXiv:2408.11686}, 2024.

\bibitem{caluya2021wasserstein}
K.~Caluya and A.~Halder, ``Wasserstein proximal algorithms for the
  {S}chr{\"o}dinger bridge problem: Density control with nonlinear drift,''
  \emph{IEEE Trans. Automatic Control}, vol.~67, no.~3, pp. 1163--1178, 2021.

\bibitem{marino2020optimal}
S.~Marino and A.~Gerolin, ``An optimal transport approach for the
  {S}chr{\"o}dinger bridge problem and convergence of sinkhorn algorithm,''
  \emph{J. Sci. Comput.}, vol. 85(2), no.~27, 2020.

\bibitem{posner1975random}
E.~Posner, ``Random coding strategies for minimum entropy,'' \emph{IEEE Trans.
  Inf. Theory}, vol.~21, no.~4, pp. 388--391, 1975.

\bibitem{Polyanskiy_Wu_2025}
Y.~Polyanskiy and Y.~Wu, \emph{Information Theory: From Coding to
  Learning}.\hskip 1em plus 0.5em minus 0.4em\relax Cambridge University Press,
  2025.

\end{thebibliography}









\newpage
\onecolumn
\appendices

\section{Proof of Lemma \ref{lemma: transform the rate-distortion problem}}
\label{appendix: proof of lemma: transform the rate-distortion problem}

For any fixed $\pi_x\in\pc(\Y)$ with $\E_{\pi}\rho\leq D$, one can always find the corresponding $\nu\in\pc(\Y)$ and $\pi\in\Pi(\mu,\nu)$, where for any Borel measurable function $f:\X\times\Y\to[0,\infty]$
    \begin{equation*}
        \begin{aligned}
          \int_{\X\times\Y} f(x,y)d\pi(x,y)&=\int_{\X}d\mu(x)\int_{\Y} f(x,y)d\pi_x(y)\\
          d\nu(y)&=\int_{\X}d\pi(x,y).
        \end{aligned}
    \end{equation*}
 Hence we can obtain that
 $$\left.I(X;Y)\right|_{\pi_x}\geqslant \inf_{\nu,\pi:\E_{\pi}\rho\leqslant D} I(X;Y)$$
 that is,
 $$\inf_{\pi_x:\E_{\pi}\rho\leqslant D}I(X;Y)\geqslant \inf_{\nu,\pi:\E_{\pi}\rho\leqslant D} I(X;Y).$$
 Additionally, for all fixed $\nu,\pi:\E_{\pi}\rho\leqslant D$, one can also determine the disintegration $\pi_x$ of $\pi$, and we have
 $$\left.I(X;Y)\right|_{\nu,\pi}\geqslant \inf_{\pi_x:\E_{\pi}\rho\leqslant D} I(X;Y)$$
 that is,
 $$\inf_{\nu,\pi:\E_{\pi}\rho\leqslant D} I(X;Y)\geqslant \inf_{\pi_x:\E_{\pi}\rho\leqslant D} I(X;Y)\Longrightarrow \inf_{\nu,\pi:\E_{\pi}\rho\leqslant D} I(X;Y)=\inf_{\pi_x:\E_{\pi}\rho\leqslant D} I(X;Y).$$
 
 Furthermore, we take the infimum sequence $\{\nu_k,\pi_k\}$ of $\inf_{\nu,\pi:\E_{\pi}\rho\leqslant D} I(X;Y)$, then
 \begin{equation*}
     \begin{aligned}
         \inf\limits_{\nu}\inf\limits_{\pi\in\Pi(\mu,\nu):\mathbb{E}_{\pi}\rho\leqslant D}I(X;Y)&\leqslant \inf\limits_{\nu}\inf\limits_{\pi\in\Pi(\mu,\nu_k):\mathbb{E}_{\pi}\rho\leqslant D}I(X;Y)\\
         &=\inf\limits_{\pi\in\Pi(\mu,\nu_k):\mathbb{E}_{\pi}\rho\leqslant D}I(X;Y)\\
         &\leqslant \left.I(X;Y)\right|_{(\nu_k,\pi_k)}.
     \end{aligned}
 \end{equation*}
 Hence 
 $$ \inf\limits_{\nu}\inf\limits_{\pi\in\Pi(\mu,\nu):\mathbb{E}_{\pi}\rho\leqslant D}I(X;Y)\leqslant \lim_{k\to\infty}\left.I(X;Y)\right|_{(\nu_k,\pi_k)}=\inf_{\nu,\pi:\E_{\pi}\rho\leqslant D} I(X;Y).$$
 Conversely, since for any fixed $\Tilde{\nu}\in\pc(\Y)$,
 $$\inf_{\nu,\pi:\E_{\pi}\rho\leqslant D} I(X;Y)\leqslant \inf_{\pi\in\Pi(\mu,\Tilde{\nu}):\E_{\pi}\rho\leqslant D} I(X;Y)$$
 we take the infimum on both sides over $\Tilde{\nu}$, and we have
 $$\inf_{\nu,\pi:\E_{\pi}\rho\leqslant D} I(X;Y)\leqslant \inf_{\Tilde{\nu}}\inf_{\pi\in\Pi(\mu,\Tilde{\nu}):\E_{\pi}\rho\leqslant D} I(X;Y)=\inf\limits_{\nu}\inf\limits_{\pi\in\Pi(\mu,\nu):\mathbb{E}_{\pi}\rho\leqslant D}I(X;Y).$$

\section{Proof of Lemma \ref{lemma: lower semicontinuity of C}}
\label{appendix: proof of lemma: lower semicontinuity of C}
\noindent \textbf{Step 1. Jointly lower semi-continuous: }

    For $(x_k,p_k)\to (x,p)$ in the sense of product topology, i.e. $x_n\to x\in \X$ and $p_n\to p\in \pc(\Y)$, we need to verify that $C(x,p)\leq \liminf\limits_{k\to\infty}C(x_k,p_k)$.
    On the one hand, we have
    \begin{equation*}
        \begin{aligned}
            \int_{\Y}\log\dfrac{dp}{d\nu} dp\leqslant \liminf_{k\to\infty}\int_{\Y}\log\dfrac{dp_k}{d\nu} dp_k
        \end{aligned}
    \end{equation*}
    by \textit{Theorem 1} in \cite{posner1975random} or \textit{Theorem 4.9} in \cite{Polyanskiy_Wu_2025}. On the other hand, $\displaystyle \beta\int_{\Y} (\rho(x,y)-D)dp$ is lower semicontinuous in $(x,p)$ by the Assumption \ref{assumption: lower semicontinuity of C}. Hence $C(x,p)$ is jointly lower semicontinuous in this case.
    \\
    \textbf{Step 2. Lower bounded: }
    \begin{equation*}
        \begin{aligned}
            C(x,p)&\coloneqq \int_{\Y}\left[\log\dfrac{dp}{d\nu}+\beta(\rho(x,y)-D)\right]dp\\
            &\geqslant \int_{\Y} \log\dfrac{dp}{d\nu}dp-\beta D\\
            &\geqslant -\beta D.
        \end{aligned}
    \end{equation*}
    Note that the relative entropy is always non-negative since
    \begin{equation*}
        \begin{aligned}
            \int_{\Y}\log\dfrac{dp}{d\nu}dp&=\int_{\Y}\varphi(\dfrac{dp}{d\nu})d\nu\\
             &\geqslant \varphi\left(\int_{\Y}\dfrac{dp}{d\nu}d\nu\right)\\
             &=\varphi(1)=0.
        \end{aligned}
    \end{equation*}
    Here $\varphi(x)=x\log x$ is a convex function over $(0,+\infty)$.
    \\
    \textbf{Step 3. Convex in the second argument: }
    
    For any $p_1,p_2\in\pc(\Y)$ and $\lambda\in[0,1]$, we have
    \begin{equation*}
        \begin{aligned}
    C&\left( x,\lambda p_1+ \left( 1-\lambda \right) p _2 \right) - \lambda C\left( x,p _1 \right) +\left( 1-\lambda \right) C\left( x,p _2 \right) \\
    &=
    \int_{\Y}\left[
        \log\dfrac{\lambda p_1+(1-\lambda p_2)}{d\nu}+\beta \rho(x,y)
    \right](\lambda dp_1+(1-\lambda)dp_2)
    \\
    &\;\; - \lambda \int_Y{\,\,}\left[ \log \frac{dp _1}{d\nu}+\beta \rho \left( x,y \right) \right] dp _1+\left( 1-\lambda \right) \int_Y{\,\,}\left[ \log \frac{dp _2}{d\nu}+\beta \rho \left( x,y \right) \right] \,\,dp _2\\
    &=
    \int_Y\log \frac{\lambda dp _1+\left( 1-\lambda \right) dp _2}{d\nu}\,\,\left( \lambda \,\,dp _1+\left( 1-\lambda \right) \,\,dp _2 \right) - \lambda \int_Y{\,\,}\log \frac{dp _1}{d\nu}\,\,dp _1+\left( 1-\lambda \right) \int_Y{\,\,}\log \frac{dp _2}{d\nu}\,\,dp _2\\
    &=\int_{\Y}\varphi(\lambda\dfrac{dp_1}{d\nu}+(1-\lambda)\dfrac{dp_2}{d\nu})d\nu-\lambda\int_{\Y}\varphi(\dfrac{dp_1}{d\nu})d\nu-(1-\lambda)\int_{\Y}\varphi(\dfrac{dp_2}{d\nu})d\nu\\
    &\leqslant 0.
        \end{aligned}
    \end{equation*}

\section{Proof of Theorem \ref{thm: refined existence}}
\label{appendix: proof of refined existence}

The lower bounded property and the convexity in the second argument of $C(x,p)$ over $\pc_t(\Y)$ can be proved by the same method as in Lemma \ref{lemma: lower semicontinuity of C}. Hence we only need to prove the jointly lowe semi-continuity of $C(x,p)$ over $\X\times\pc_t(\Y)$. On the one hand, we have
\begin{equation*}
    \begin{aligned}
        \int_{\Y}\log\dfrac{dp}{d\nu} dp\leqslant \liminf_{k\to\infty}\int_{\Y}\log\dfrac{dp_k}{d\nu} dp_k
    \end{aligned}
\end{equation*}
by \textit{Theorem 1} in \cite{posner1975random} or \textit{Theorem 4.9} in \cite{Polyanskiy_Wu_2025}. On the other hand, $\displaystyle \beta\int_{\Y} (\rho(x,y)-D)dp$ is lower semicontinuous in $(x,p)\in \X\times\pc_t(\Y)$ by the Assumption \ref{assumption: refined lower semicontinuity of C}. Hence $C(x,p)$ is jointly lower semicontinuous in this case. 

Furthermore, by Theorem \ref{thm:existence}, we know that the minimizer $\pi^\star$ of $J(\nu,\beta)$ exists.

\section{Proof of Proposition \ref{proposition: the target measure has finite t-Wasserstein distance}}
\label{appendix: proof of proposition 1}

By the property of Wasserstein distance, we have
\begin{equation*}
    \begin{aligned}
        W_t(\nu,\delta_{y_0})&=\left(\int_{\X}d_{\X}(x,y_0)^t\right)^{1/t}\\ 
        &\leqslant W_t(\mu,\delta_{y_0})+W_t(\mu,\nu).
    \end{aligned}
\end{equation*}
Note that
\begin{equation*}
    \begin{aligned}
        W_t(\mu,\nu)&=\inf_{\pi\in\Pi(\mu,\nu)}\left[\E_{\pi}(d_{X}^t)\right]^{1/t}\\ 
        &\leqslant c^{-1/t}  \inf_{\pi\in\Pi(\mu,\nu)}\left[\E_{\pi}\rho\right]^{1/t}\\
        &\leqslant c^{-1/t} D^{1/t}.
    \end{aligned}
\end{equation*}
Hence we have
$$W_t(\nu,\delta_{y_0})\leqslant c^{-1/t}D^{1/t}+W_t(\mu,\delta_{y_0})<+\infty.$$

\section{Proof of Theorem \ref{thm: parametric representation-1}}
\label{appendix: proof of thm: parametric representation-1}

By Lemma \ref{lemma: explicit expression}, we know that the minimizers $\pi_x^{\star},\pi^{\star}$ exist and there exist nonnegative measurable functions $f,g$ such that
$$\dfrac{d\pi^\star}{d\gamma}=f(x)g(y), \gamma-a.e..$$
Namely,
$$d\pi^{\star}=Kf(x)g(y)e^{-\beta \rho(x,y)}d\mu(x)d\nu(y).$$
Then we have:
\begin{equation*}
	\begin{aligned}
    d\pi_x^{\star}&=Kf(x)g(y)e^{-\beta\rho(x,y)}d\nu=\dfrac{K}{K_x}f(x)g(y)d\gamma_x\\
    \dfrac{d\pi_x^\star}{d\gamma_x}&=\dfrac{Kf(x)g(y)}{K_x}.
\end{aligned}
\end{equation*}
Since $\pi^\star\in \Pi(\mu,\nu)$, we have $\int_{\Y}d\pi^\star=d\mu, \int_{\X}d\pi^\star=d\nu$, which can be written as Schr\"odinger equations\cite{nutz2021introduction}
\begin{equation}
    \begin{aligned}
       K&=\left(\int_{\X\times\Y}e^{-\beta\rho(x,y)}d\mu d\nu\right)^{-1}=\left(\int_{\X\times\Y}f(x)g(y)e^{-\beta\rho(x,y)}d\mu d\nu\right)^{-1}\\ 
       1&=Kf(x)\int_{\Y}g(y)e^{-\beta\rho(x,y)}d\nu=Kg(y)\int_{\X}f(x)e^{-\beta\rho(x,y)}d\mu \\ 
       \Longrightarrow& \int_{\X}\dfrac{g(y')e^{-\beta\rho(x,y')}}{\int_{\Y}g(y)e^{-\beta\rho(x,y)}d\nu}d\mu =1,\quad\text{$\forall y'\in\Y$}.
    \end{aligned}
\end{equation}
Then we have
\begin{equation}\label{eq: explicit expression of J(v,b)}
    \begin{aligned}
        J(\nu,\beta)&=\int_{\X\times\Y}\log\dfrac{d\pi^\star}{d\gamma}d\pi^\star+\log K-\beta D\\ 
        &=\int_{\X\times\Y} \log (f(x)g(y))d\pi^\star+\log K -\beta D\\ 
        &=\int_{\X} \log (Kf(x))d\mu+\int_{\Y}\log(g(y)) d\nu-\beta D\\
        &=-\int_{\X}\log\left(\int_{\Y}g(y)e^{-\beta\rho(x,y)}d\nu\right)d\mu+\int_{\Y}\log g(y)d\nu-\beta D\\ 
        &=-\int_{\X}\log\left(\int_{\Y}e^{-\beta\rho(x,y)}d\nu\right)d\mu-\beta D +\int_{\X}\log\left(\dfrac{\int_{\Y}e^{-\beta\rho(x,y)}d\nu}{\int_{\Y}g(y)e^{-\beta\rho(x,y)}d\nu}\right)d\mu +\int_{\Y}\log g(y)d\nu
    \end{aligned}
\end{equation}
\begin{equation*}
	\begin{aligned}
		R(D)
		&=\inf\limits_{\nu}\inf\limits_{\pi\in\Pi(\mu,\nu):\mathbb{E}_{\pi}\rho\leqslant D}I(X;Y)\\
		&\geqslant 
		\inf\limits_{\nu}\sup_{\beta\geqslant 0}\inf\limits_{\pi\in\Pi(\mu,\nu)}I(X;Y)+\beta(\mathbb{E}_{\pi}\rho-D)\\
		&=\inf\limits_{\nu}\sup_{\beta\geqslant 0}
		\left\{
			-\int_{\X}\log\left(\int_{\Y}e^{-\beta\rho(x,y)}d\nu\right)d\mu-\beta D+\int_{\X}\log\left(\dfrac{\int_{\Y}e^{-\beta\rho(x,y)}d\nu}{\int_{\Y}g(y)e^{-\beta\rho(x,y)}d\nu}\right)d\mu +\int_{\Y}\log g(y)d\nu
		\right\}\\ 
		&\geqslant 
		\sup_{\beta\geqslant 0}\inf\limits_{\nu}
		\left\{
			-\int_{\X}\log\left(\int_{\Y}e^{-\beta\rho(x,y)}d\nu\right)d\mu-\beta D+\int_{\X}\log\left(\dfrac{\int_{\Y}e^{-\beta\rho(x,y)}d\nu}{\int_{\Y}g(y)e^{-\beta\rho(x,y)}d\nu}\right)d\mu +\int_{\Y}\log g(y)d\nu
		\right\}.
	\end{aligned}
\end{equation*}
Define
$$L(\beta)=\inf\limits_{\nu}
\left\{
    -\int_{\X}\log\left(\int_{\Y}e^{-\beta\rho(x,y)}d\nu\right)d\mu+\int_{\X}\log\left(\dfrac{\int_{\Y}e^{-\beta\rho(x,y)}d\nu}{\int_{\Y}g(y)e^{-\beta\rho(x,y)}d\nu}\right)d\mu +\int_{\Y}\log g(y)d\nu
\right\}.$$
Then for any $\beta\geqslant 0$
\begin{equation}
    \begin{aligned}
        R(D)&\geqslant L(\beta)-\beta D\\ 
        R(D)+\beta D&\geqslant L(\beta).
    \end{aligned}
\end{equation}
It is well known that $R(D)$ is a non-increasing convex function of $D$, hence it has second derivatives almost everywhere, and for any $D\geqslant D_{\min}$ there exists subgradient of $R(D)$, denoted by $\partial R(D)$, which is always a closed, convex and bounded set for convex function $R(D)$. Then for each $-\beta_{D}\in\partial R(D)$, we have
$$R(D')\geqslant R(D)-\beta_D (D'-D), \;\;\forall D'\in (D_{\min}, D_{\max})$$
that is,
$$R(D')+\beta_D D'\geqslant R(D)+\beta_D D\geqslant L(\beta_D).$$
Additionally, by the definition of $L(\beta)$, for any sufficiently small $\epsilon>0$, there exists corresponding $\nu^{\epsilon}\in\pc(\Y)$, such that
$$-\int_{\X}\log\left(\int_{\Y}e^{-\beta_D\rho(x,y)}d\nu^{\epsilon}\right)d\mu+\int_{\X}\log\left(\dfrac{\int_{\Y}e^{-\beta_D\rho(x,y)}d\nu^{\epsilon}}{\int_{\Y}g(y)e^{-\beta_D\rho(x,y)}d\nu^{\epsilon}}\right)d\mu +\int_{\Y}\log g(y)d\nu^{\epsilon}<L(\beta_D)+\epsilon.$$
On the other hand, 
\begin{equation}
\begin{aligned}
    -\int_{\X}\log\left(\int_{\Y}e^{-\beta_D\rho(x,y)}d\nu^{\epsilon}\right)d\mu+&\int_{\X}\log\left(\dfrac{\int_{\Y}e^{-\beta_D\rho(x,y)}d\nu^{\epsilon}}{\int_{\Y}g(y)e^{-\beta_D\rho(x,y)}d\nu^{\epsilon}}\right)d\mu +\int_{\Y}\log g(y)d\nu^{\epsilon}\\ 
    &=\left.\left(I(X;Y)+\beta_D \E_{\pi}\rho\right)\right|_{\pi_x=\pi_x^{\star}(\nu^\epsilon)}\\ 
    &\geqslant R(\E_{\pi^\star(\nu^\epsilon)}\rho)+\beta_D \E_{\pi^\star(\nu^\epsilon)}\rho\\ 
    &\geqslant R(D)+\beta_D D.
\end{aligned}
\end{equation}
By equations above, we have
$$R(D)+\beta_D D-\epsilon <L(\beta_D)\leqslant R(D)+\beta_D D, \quad \forall \epsilon>0.$$
Furthermore, 
\begin{equation}
    \begin{aligned}
        L(\beta_D)&=R(D)+\beta_D D\\ 
        R(D)&\geqslant \sup_{\beta\geqslant 0}(L(\beta)-\beta D)\\ 
            &\geqslant \sup_{\beta_D\in\partial R(D)}(L(\beta_D)-\beta_D D)\\ 
            &=R(D).
    \end{aligned}
\end{equation}
\noindent Above all, we have proved that
\begin{equation}
	\begin{aligned}
R(D)&=\inf\limits_{\pi_x:\mathbb{E}_{\pi}\rho\leqslant D}I(X;Y)
=\inf\limits_{\nu}\inf\limits_{\pi\in\Pi(\mu,\nu):\mathbb{E}_{\pi}\rho\leqslant D}I(X;Y) 
=\inf\limits_{\nu,\pi\in\Pi(\mu,\nu):\mathbb{E}_{\pi}\rho\leqslant D}I(X;Y)\\ 
&=\inf\limits_{\nu}\sup_{\beta\geqslant 0}\inf\limits_{\pi\in\Pi(\mu,\nu)}I(X;Y)+\beta(\mathbb{E}_{\pi}\rho-D)\\
&=\inf\limits_{\nu}\sup_{\beta\geqslant 0}
\left\{
	-\int_{\X}\log\left(\int_{\Y}e^{-\beta\rho(x,y)}d\nu\right)d\mu-\beta D +\int_{\X}\log\left(\dfrac{\int_{\Y}e^{-\beta\rho(x,y)}d\nu}{\int_{\Y}g(y)e^{-\beta\rho(x,y)}d\nu}\right)d\mu +\int_{\Y}\log g(y)d\nu
\right\}\\ 
&=
\sup_{\beta\geqslant 0}\inf\limits_{\nu}
\left\{
	-\int_{\X}\log\left(\int_{\Y}e^{-\beta\rho(x,y)}d\nu\right)d\mu-\beta D+\int_{\X}\log\left(\dfrac{\int_{\Y}e^{-\beta\rho(x,y)}d\nu}{\int_{\Y}g(y)e^{-\beta\rho(x,y)}d\nu}\right)d\mu +\int_{\Y}\log g(y)d\nu
\right\}\\ 
&=\inf\limits_{\nu}\max_{\beta\in \partial R(D)}
\left\{
	-\int_{\X}\log\left(\int_{\Y}e^{-\beta\rho(x,y)}d\nu\right)d\mu-\beta D+\int_{\X}\log\left(\dfrac{\int_{\Y}e^{-\beta\rho(x,y)}d\nu}{\int_{\Y}g(y)e^{-\beta\rho(x,y)}d\nu}\right)d\mu +\int_{\Y}\log g(y)d\nu
\right\}\\ 
&=
\max_{\beta\in \partial R(D)}\inf\limits_{\nu}
\left\{
	-\int_{\X}\log\left(\int_{\Y}e^{-\beta\rho(x,y)}d\nu\right)d\mu-\beta D+\int_{\X}\log\left(\dfrac{\int_{\Y}e^{-\beta\rho(x,y)}d\nu}{\int_{\Y}g(y)e^{-\beta\rho(x,y)}d\nu}\right)d\mu +\int_{\Y}\log g(y)d\nu
\right\}\\ 
&=\inf\limits_{\nu}
\left\{
	-\int_{\X}\log\left(\int_{\Y}e^{-\beta\rho(x,y)}d\nu\right)d\mu-\beta D+L(\nu,\beta)
\right\},\;\;\forall \beta\in\partial R(D)
	\end{aligned}
\end{equation}
where 
$$L(\nu,\beta)=\int_{\X}\log\left(\dfrac{\int_{\Y}e^{-\beta\rho(x,y)}d\nu}{\int_{\Y}g(y)e^{-\beta\rho(x,y)}d\nu}\right)d\mu +\int_{\Y}\log g(y)d\nu.$$

\section{Proof of Theorem \ref{thm: parametric representation-2}}
\label{appendix: proof of thm: parametric representation-2}

We have known that for each $D>D_{\min}$ and each $\beta\in \partial R(D)$
\begin{equation}
\begin{aligned}
R(D)&=\inf_{\nu\in\pc_t(\Y)}\inf_{\pi\in\Pi(\mu,\nu)}\left\{ 
    \int_{\X\times\Y} \log\dfrac{d\pi}{d\mu\times d\nu} d\pi +\beta \int_{\X\times\Y}\rho(x,y)d\pi 
\right\}-\beta D.
\end{aligned}
\end{equation}
Define
\begin{equation}
    \begin{aligned}
        \Tilde{J}(\pi,\nu,\beta)&=D_{KL}(\pi\| \mu\times \nu)+\beta\E_{\pi}(\rho)=\int_{\X\times\Y}\log\dfrac{d\pi}{d\mu\times d\nu}d\pi+\beta\int_{\X\times\Y}\rho(x,y)d\pi\\ 
        f_{\nu,\beta}(x)&=\left( 
            \int_{\Y}e^{-\beta\rho(x,y)}d\nu
        \right)^{-1}\\ 
        \pi_{\nu,\beta}(dx,dy)&=f_{\nu,\beta}(x)e^{-\beta\rho(x,y)}\mu(dx)\times \nu(dy).
    \end{aligned}
    \end{equation}
Then we obtain
\begin{equation}
\begin{aligned}
R(D)&=\inf_{\nu}\inf_{\pi\in\Pi(\mu,\nu)} \Tilde{J}(\pi,\nu,\beta)-\beta D\\ 
&=\inf_{\pi\in\Pi(\mu,\cdot)} \Tilde{J}(\pi,\pi_1,\beta)-\beta D,\;\;\text{where $\pi_1$ is the $Y$-marginal of $\pi$.}
\end{aligned}
\end{equation}
Assume that the optimal reconstruction $\nu^{\star}$ exists, then we have
\begin{equation}
    \begin{aligned}
    R(D)&=\inf_{\pi\in\Pi(\mu,\nu^{\star})} \Tilde{J}(\pi,\nu^{\star},\beta)-\beta D\\ 
    &\xlongequal{OWT} \Tilde{J}(\pi^{\star},\nu^{\star},\beta)-\beta D.
    \end{aligned}
\end{equation}
On the other hand, 
\begin{equation}
    \begin{aligned}
    R(D)
    &=\inf_{\pi\in\Pi(\mu,\cdot)} \Tilde{J}(\pi, \pi_1, \beta)-\beta D,\;\;\text{where $\pi_1$ is the $Y$-marginal of $\pi$.}
    \end{aligned}
\end{equation}
If $\inf_{\pi\in\Pi(\mu,\cdot)} \Tilde{J}(\pi, \pi_1, \beta)$ can not be achieved, then there exists $\Tilde{\pi}$, such that 
\begin{equation}
    \begin{aligned}
    \Tilde{J}(\Tilde{\pi}, \Tilde{\pi}_1, \beta)<\Tilde{J}(\pi^{\star},\nu^{\star},\beta).
    \end{aligned}
\end{equation} 
Then we have
\begin{equation}
    \begin{aligned}
    R(D)+\beta D &= \Tilde{J}(\pi^\star,\nu^\star,\beta)> \Tilde{J}(\Tilde{\pi}, \Tilde{\pi}_1, \beta)\\ 
    &\geqslant \inf_{\pi\in\Pi(\mu,\cdot)} \Tilde{J}(\pi, \pi_1, \beta)\\ 
    &=R(D)+\beta D.
    \end{aligned}
\end{equation}
It is a contradiction. Hence $\inf_{\pi\in\Pi(\mu,\cdot)} \Tilde{J}(\pi, \pi_1, \beta)$ can be achieved at $\pi^\star,\nu^\star$ if the optimal reconstruction $\nu^{\star}$ exists. Furthermore, by Lemma 1.4 in \cite{csiszar1974extremum}, we have
\begin{equation}
    \begin{aligned}
    d\pi^\star &= \alpha(x)e^{-\beta\rho(x,y)}d\mu\times d\nu^\star\\ 
    &=\dfrac{e^{-\beta\rho(x,y)}}{\int_{\Y}e^{-\beta\rho(x,y)}d\nu^\star}d\mu\times d\nu^\star.
    \end{aligned}
\end{equation}

\end{document}